\documentclass[runningheads]{llncs}
\usepackage[T1]{fontenc}

\usepackage[T1]{fontenc}
\usepackage{amssymb, mathtools}
\usepackage[mathlines]{lineno}

\usepackage[colorinlistoftodos]{todonotes}
\usepackage[nolist]{acronym}
\usepackage{arydshln}

\renewcommand{\epsilon}{\varepsilon}
\renewcommand{\phi}{\varphi}

\newcommand{\taumax}{\tau_{\max}}

\newcommand*\diff{\mathop{}\!\mathrm{d}}

\usepackage[edges]{forest}		
\usepackage{tikz}
\usetikzlibrary{arrows.meta}
\usepackage{tikz-3dplot}
\usetikzlibrary{calc,arrows,automata,patterns,graphs,shapes,petri,decorations.pathmorphing,decorations.markings,decorations.pathreplacing} 
\tikzset{->-/.style={decoration={markings,mark=at position .5 with {\arrow{>}}},postaction={decorate}}}
\tikzset{vertex/.style={draw,circle,inner sep=0pt,minimum size=18pt},>=latex'}
\tikzset{c/.style={draw,circle,inner sep=0pt,minimum size=15pt},>=latex'}
\tikzset{dot/.style={draw,circle,inner sep=0pt,minimum size=3pt, fill=black},>=latex'}

\usepackage{cleveref}
\usepackage{optidef}

\makeatletter
\let\c@lemma\c@theorem

\let\c@claim\c@theorem

\let\c@corollary\c@theorem

\let\c@observation\c@theorem

\let\c@definition\c@theorem

\makeatother

\crefname{theorem}{theorem}{theorems}
\Crefname{theorem}{Theorem}{Theorems}

\crefname{lemma}{lemma}{lemmas}
\Crefname{lemma}{Lemma}{Lemmas}

\crefname{claim}{claim}{claims}
\Crefname{claim}{Claim}{Claims}

\crefname{corollary}{corollary}{corollaries}
\Crefname{corollary}{Corollary}{Corollaries}

\crefname{observation}{observation}{observations}
\Crefname{observation}{Observation}{Observations}

\crefname{definition}{definition}{definitions}
\Crefname{definition}{Definition}{Definitions}

\title{Nash Flows Over Time with Tolls} 

\author{%
Shaul Rosner\inst{1}\orcidID{0009-0006-4671-2566}\thanks{Supported by the European Research Council (ERC) under the European Union's Horizon 2020 research and innovation program (grant agreements No. 101170373, No.101077862)} \and
Marc Schr\"oder\inst{2}\orcidID{0000-0002-0048-2826} \and
Laura Vargas Koch\inst{3}\orcidID{0000-0002-7499-5958}%
}

\institute{
Tel Aviv University, Ramat Aviv, Tel Aviv 6997801, Israel\\
\email{srosner@tauex.tau.ac.il} \and
Maastricht University, Tongersestraat 53, 6211 KL Maastricht, The Netherlands\\
\email{m.schroder@maastrichtuniversity.nl} \and
RWTH Aachen University, Kackertstr. 7, 52072 Aachen, Germany\\
\email{vargaskoch@gdm.rwth-aachen.de}
}
\authorrunning{S. Rosner et al.} 

\begin{document}
\maketitle

\begin{abstract}
We study a dynamic routing game motivated by traffic flows. 
The base model for an edge is the Vickrey bottleneck model. That is, edges are equipped with a free flow transit time and a capacity. When the inflow into an edge exceeds its capacity, a queue forms and the following particles experience a waiting time. 
In this paper, we enhance the model by introducing tolls, i.e., a cost each flow particle must pay for traversing an edge.
In this setting we consider non-atomic equilibria, which means flows over time in which every particle is on a cheapest path, when summing up toll and travel time.

We first show that unlike in the non-tolled version of this model, dynamic equilibria are not unique in terms of costs and do not necessarily reach a steady state. As a main result, we provide a procedure to compute steady states in the model with tolls.
\end{abstract}

\section{Introduction}

In a classic paper, Ford and Fulkerson \cite{ford1958constructing} introduced dynamic flows in order to capture non-static behavior of flow systems as is apparent in for example traffic flows. Modeling wise, this means that edges are equipped with a transit time and a capacity, and flow moves over time through the network (and is described by time dependent functions of flow rates). 
Such dynamic flows or flows over time have been well studied, both from an optimization as well as from a game-theoretic perspective.

In the game-theoretic models, the capacity bounds the rate with which flow can exit an edge. If the flow entering an edge exceeds the capacity, a queue forms, and additional flow queues up before leaving the edge. Such edge behavior is denoted as the Vickrey bottleneck model and was introduced already by Vickrey \cite{vickrey1969congestion} in 1969. An equilibrium in this model is denoted as a dynamic equilibrium or Nash flow over time and has been introduced by Koch and Skutella \cite{koch2011nash}. Based on this work, Cominetti, Correa and Larr\'e \cite{CCL15} showed existence of equilibrium states. Moreover, in these works it was shown that an equilibrium, i.e., a Nash flow over time, can be constructed by piecewise linear extension in so called phases. 
Cominetti, Correa and Olver \cite{CCO22} showed that if the inflow does not exceed the minimum cut of the network, all queues remain bounded and there is some point in time, from which onward queues no longer change. In other words, after a finite time a Nash flow over time reaches a phase, that lasts forever. We denote this final phase as a {\em steady state}.
Olver, Sering and Vargas Koch \cite{OSK22} extended the result to the case where the inflow is larger than the capacity of the minimum cut. Clearly, for this case queues cannot be expected to remain constant. Instead, they showed that from some finite point in time onward, each queue either remains constant, or grows at a constant rate.

We augment the standard flows over time model by introducing constant tolls. In our model, every edge has a transit time, a capacity and a toll. Particles in this model aim to minimize their travel time plus their toll, which is denoted as their cost. Note that travel time and toll do not have the same impact on a particle, as travel time leads to a delay in arrival time at later edges, and thus results in different waiting times on these edges.
While tolling questions are well-studied in the static setting, see e.g., \cite{pigou,beckmann,cole,gairing}, in the dynamic setting little is known so far.
Olver and Frascaria \cite{frascaria2020algorithms} showed that in a setting for which particles may choose their departure time, there exist tolls
implementing an earliest arrival flow. Implementing here means that the flow is an equilibrium state in a network with the given tolls.
Recently, Graf, Harks, and Schwarz \cite{graf2025tolls} have characterized all flows that can be implemented by tolls. 
Moreover, they have shown that system optimal flows can be implemented by tolls (see \cite{graf2025system}). 
Note, that their definition of tolls differs from ours in two ways. First, their tolls are time-dependent and are even allowed to be non-continuous and second, particles are split into commodities and value time differently. 
A different interpretation of costs is given by Oosterwijk et al. \cite{OSS} in which the primary objective of particles is to be on time, but the secondary objective is to minimize costs.

In a way, our analysis of tolls takes the opposite perspective. Instead of starting with a given flow, we consider simple, constant tolls and aim to understand how resulting equilibrium flows behave. In the single-commodity case, this turns out to be much more challenging than in the non-tolled setting. In the non-tolled setting, a strict first-in first-out (FIFO) property holds. Therefore, in an equilibrium state a flow particle that enters the network is only influenced by flow that entered the network at an earlier time, and not by flow that enters the network later than the particular particle.
This is no longer the case in the tolled setting, as in certain cases it is preferable for some particles to travel via a longer path with a low toll, resulting in the particles queueing at an edge after particles which entered the network at a later time, but used a shorter path with a higher toll. We refer to \Cref{sec:examples} for a concrete example. 

This induces problems that are similar to problems in multi-commodity instances without tolls. 
For multi-commodity instances the characterization of equilibrium flows is still wide open. Existence is known \cite{CCL15}, but the structure of equilibrium states is much less understood. Only for special cases, where FIFO is still fulfilled, equilibria can be characterized in piece-wise linear phases as in the single-commodity setting, see Iryo and Smith \cite{iryo2017uniqueness}. Moreover, nice properties such as uniqueness of equilibria do not hold any longer
\cite{iryo2011multiple}.
Also in the extension to spillback, properties as uniqueness of equilibria \cite{sering2019nash} and convergence to a steady state \cite{ziemke2023spillback} do not hold any longer. However, here the FIFO property prevails.
This said, it is not surprising that equilibrium states in a setting with tolls are much more challenging to analyze and lack some of the nice properties that Nash flows over time (in the non-tolled setting) fulfill.

\subsection{Our Contribution}
 As our main result, we give a procedure to compute steady states in the Nash flows over time with tolls setting.
 More concretely, we provide a linear program and its dual which characterize steady states for Nash flows over time. This LP is a modified version of the LP introduced by Cominetti, Correa and Olver \cite{CCO22} and then adapted in \cite{OSK22}. However, showing that the state described by the LP is indeed an equilibrium state turns out to be more challenging in the toll setting. 
 For the case where the capacity of the minimum cut is greater or equal to the inflow capacity in the steady state, the difference between toll and transit time resolves and the setting reduces to the standard Nash flow over time setting. However, if the inflow exceeds the capacities and queues continue to grow, there is an interesting trade-off between tolls and transit-times which depends on the growth rate of the queues. 

The LP determines the initial queues and the inflow in a steady state. To make this correspond to an equilibrium, we need to allow for some short phase, where the network fills up with flow. We do so by introducing deficit flows, in which all vertices are treated as a source for a short time period. In this way the network fills up quickly, and after the maximal transit time, we obtain a proper flow which satisfies flow conservation and the equilibrium constraints.

 The LP is based on a specific thin flow solution. As part of our discussion of thin flows, we also provide an alternative definition of thin flows from a sink-based perspective, which is very natural in this setting. This definition is independent of tolls and thus also applicable in the non-tolled model.

As our second main contribution, we present two instances. The first example shows the non-uniqueness of costs of Nash flows over time with tolls, while uniqueness was shown to hold true for Nash flows over time without tolls \cite{OSK22}. 
The second example shows that Nash flows over time with tolls do not reach a steady state after finite time, again in contrast to the behavior of Nash flows over time.

 \subsection{Structure of the paper}
After introducing the model in \Cref{sec:model}, we discuss two interesting examples in \Cref{sec:examples}. These show that crucial properties of standard Nash flows over time do not hold any longer in the setting with tolls. In particular, we cannot hope to show that every equilibrium trajectory reaches a steady state. 
However, in \Cref{sec:steady_state}, we show how to compute a family of steady states. This means that if we determine initial queues, and fill the network with a pre-flow for a finite time period, the resulting state is a steady state, in which the queues all change with a linear rate.
We conclude the paper with a discussion in \Cref{sec:discussion}.

\section{Model}
\label{sec:model}
Let $G=(V,E)$ be a directed graph with two special vertices, a source $s\in V$ and a sink $t\in V$. We assume that every vertex is reachable from $s$, and $t$ is reachable from every vertex. This is w.l.o.g. as all other vertices can simply be removed. Each edge $e\in E$ has a capacity $\nu_e > 0$, a transit time $\tau_e \geq 0$ and a fixed toll $c_e \geq 0$.
Flow has to travel from $s$ to $t$ and departs from $s$ at time $\theta$ with a constant inflow rate $u > 0$.
We call the tuple $(G=(V,E), u, (\tau_e)_{e \in E}, (\nu_e)_{e \in E}, (c_e)_{e \in E})$ with properties as defined above a \emph{tolled flow over time instance}.
We denote by $\delta^-(v)$ and $\delta^+(v)$ respectively the incoming and outgoing edges of vertex $v$, and by $\delta(v)$ the union of the two. Similarly, $\delta^-(S)$ and $\delta^+(S)$ are the sets of edges entering and leaving a set of vertices $S$ respectively, and $\delta(S)$ is their union.

A \emph{flow over time} is described by a set of edge inflow and outflow rates $f=(f_e^+,f_e^-)_{e\in E}$. We assume that the edge inflow and outflow rates satisfy flow conservation for all $v\in V$ and for
almost all $\theta\geq 0$, i.e.,
\[\sum_{e\in\delta^+(v)}f_e^+(\theta)-\sum_{e \in \delta^-(v)}f_e^-(\theta)=\begin{cases}
u &\text{ if }v=s \; ,\\
0 &\text{ if }v\neq s,t \; ,\\
\leq 0 & \text{ if } v=t \; .
\end{cases}\] 
We define the cumulative inflow and outflow function of an edge as
\[
F_e^+(\theta) \coloneqq\int_0^\theta f_e^+(\xi) \diff \xi \;, \quad
F_e^-(\theta) \coloneqq \int_0^\theta f_e^-(\xi) \diff \xi \; .
\]
If the inflow rate $f_e^+(\theta)$ is larger than the edge capacity $\nu_e$, a queue will grow. We denote the mass of this queue at time $\theta$ by $z_e(\theta) = F_e^+(\theta) - F_e^-(\theta+\tau_e)$. We assume that queues operate at capacity and according to the first-in first-out (FIFO) principle, i.e.,
\[f_e^-(\theta+\tau_e)=\begin{cases}
\nu_e &\text{ if }z_e(\theta)>0 \; ,\\
\min\{f_e^+(\theta),\nu_e\} &\text{ if }z_e(\theta)=0 \;.
\end{cases}\] which implies that the queueing time on edge $e$ at time $\theta$ is $q_e(\theta)=z_e(\theta)/\nu_e$. A particle which enters the edge at time $\theta$ therefore leaves the edge at time $T_e(\theta) = \theta + q_e(\theta) + \tau_e$.

Given a flow over time $f$, define $c_v(\theta)$ as the \emph{minimum cost to reach $t$} starting from $v$ at time $\theta$. Then $c_v(\theta)\leq \tau_e+q_e(\theta)+c_e+c_w(\theta+\tau_e+q_e(\theta))$ for all $e=vw\in E$. Moreover, there exists at least one outgoing edge for which this inequality is tight. Therefore, we define
$$c_v(\theta)\coloneqq\begin{cases}
0&\text{ if }v=t\\
\min\limits_{e=vw\in E}\tau_e+q_e(\theta)+c_e+c_w(\theta+\tau_e+q_e(\theta))&\text{ otherwise.}
\end{cases}$$ 
We call edge $e=vw$ \emph{active} at time $\theta$ if it attains the above minimum. The set of active edges at time $\theta$ is denoted by $E_{\theta}$, i.e.,  \[
E_{\theta} \coloneqq \{ e=vw \in E \mid c_v(\theta)=\tau_e+q_e(\theta)+c_e+c_w(\theta+\tau_e+q_e(\theta)) \}.
\]

Note, that the collection of functions $(c_v)_{v \in V}$ replaces the earliest arrival labels, which we also refer to as ``$\ell$-labels''. 
The earliest arrival labels where introduced by \cite{koch2011nash} and are since then used to describe Nash flows over time. 
A crucial difference is that instead of considering a particle which enters the source at time $\theta$ and asking what is the earliest time $\ell_v(\theta)$ this particle can reach vertex $v$, here we consider all the particles that reached vertex $v$ by time $\theta$ and their remaining required cost, i.e., travel time plus toll, in order to reach $t$. Note that the particles that reach $v$ at time $\theta$ may have spent a different amount of time in the network and paid a different toll to arrive at $v$. However, their remaining cost is necessarily the same. We emphasize this is a significant difficulty in using the standard earliest arrival labels for a tolled model -- as there is no unique arrival time or cost of particles in vertex $v$ at time $\theta$.

Note that even if $c_e=0$ for every edge $e \in E$, the values of $c_v$ are not equal to the values of $\ell_v$ according to the standard notation, as the $\ell$ labels are generally used to describe the prefix of the path, while the $c$ labels describe the suffix of the path.
\begin{definition}
A flow over time $f$ is called a \emph{dynamic equilibrium (DE) with tolls} if for all $e\in E$,
	$f^+_e(\theta) > 0 $ implies $ e\in E_{\theta}$ for almost all $\theta\geq 0$. 
    
\end{definition}

Existence of dynamic equilibria in settings with a fixed time horizon is given by Graf and Harks \cite[Theorem 3.5]{GH23}. Note that our model fulfills the assumptions necessary to apply their result as it inherits the continuity properties of the Vickrey bottleneck model, and our tolls are an additive fixed cost.

In this work, we study steady states of Nash flows over time with tolls.
\begin{definition}
    A flow over time $f$ is in steady state from time $T$ onward if $q_e'(\theta)$ is constant for all $e\in E$ and all $\theta \geq T$.
\end{definition}

This definition ensures stable queue growth over time. While the actual flow values may fluctuate between unsaturated edges, the experienced cost over different paths remains unchanged. In \cite{OSK22}, an equilibrium is denoted to be in weak steady state, when queues just change linearly, while the strong steady state definitions asks that all costs (earliest arrival labels) develop linearly.
For our results, on the positive side, we provide a  constructive proof yielding a constant flow, which is in steady state in the strongest possible sense (after some time). In the negative example, we provide an example where the derivatives of the queues are not constant and thus, even in the weak sense no steady state is reached. For convenience we define the notion of steady state for any flow and not just for equilibria as it was done in the literature. 

In the following, we define a flow using constant inflow rate functions. Clearly, such a flow is in steady state as queues can only be linear.

\subsection{Thin Flows}

Thin flows with resetting are a crucial component in the characterization of Nash flows over time. They were introduced by Koch and Skutella \cite{koch2011nash} and their existence was proven by Cominetti, Correa and Larr\'e \cite{CCL15}.
Thin flows over time with resetting are used to characterize the linear change of the earliest arrival time labels at a certain point in time for Nash flows over time. 
However, for characterizing the linear behavior in a steady state, a special thin flow with resetting is considered. Namely, the set of considered resetting edges is empty. 
We will focus on this special subclass and denote them as thin flows (as no resetting is happening).

Recall that in contrast to Nash flows over time our characterization of equilibria and our definition of label functions is done from the sink instead of the source. For this reason we will first define a thin flow in the classical way for the special case we need and afterwards provide a definition for a thin flow from a sink perspective. The latter is more natural in our setting while the first follows the lines of the literature on Nash flows over time. Note that thin flows from a sink perspective are independent of tolls.

\begin{definition}
  Let $(y,\lambda)$ be pair with $y \in \mathbb{R}_{\geq 0}^E$ and $\lambda \in \mathbb{R}_{\geq 0}^V$. We say $(y,\lambda)$ is a source thin flow if $y$ is an $s$-$t$-flow of value $u$ and 
\begin{align}
    \lambda_s&=1, \tag{s-TF1} \label{eq:stf1}\\\label{eq:stf2}
    \lambda_w&=\min_{e=vw\in E}\max\{\lambda_v,y_e/\nu_e\}\text{ for all }w\in V\setminus\{s\}, \tag{s-TF2}\\
    \lambda_w&=\max\{\lambda_v,y_e/\nu_e\}\text{ for all }e=vw\in E\text{ with }y_e>0. \tag{s-TF3} \label{eq:stf3}
\end{align}  
\end{definition}

In the following when referring to a thin flow, we refer to a source thin flow.

\begin{theorem}[\cite{Koc12}]
\label{thm:existence_thin_flows}
 A source thin flow always exists and can be computed in polynomial time. 
\end{theorem}

This separation will be useful for proving the following result, which we make use of in later proofs. 

\begin{lemma}\label{lem:lam}
   Let $(y,\lambda)$ be a source thin flow. Then $\lambda_s\leq \lambda_v \leq \lambda_t$ for all $v\in V$. 
\end{lemma}
\begin{proof}
    First, consider all $v\in V^+$. By \eqref{eq:stf3}, we have that $\lambda_v\leq \lambda_w$ for all $e=vw\in E$ with $y_e>0$. In particular, by moving along a flow-carrying path, we conclude that $\lambda_s\leq \lambda_v\leq \lambda_t$ for all $v\in V^+$. 

Second, consider all $v\in V\setminus V^+$. Given that all $v\in V\setminus V^+$ are reachable from $s$, there exists a path from $s$ to $v$. Consider the last vertex $u\in V^+$ on this path. Then since $u \in V^+$ we have $\lambda_u \leq \lambda_t$, and by \eqref{eq:stf2} $\lambda_v \leq \lambda_u$, so $\lambda_v\leq\lambda_u\leq \lambda_t$. Given that $t$ is reachable from all $v\in V\setminus V^+$, there exists a path from $v$ to $t$. Consider the first vertex $w\in V^+$ on this path. Then since $w \in V^+$ we have $\lambda_s \leq \lambda_w$, and by \eqref{eq:stf2} $\lambda_w \leq \lambda_v$, so  $\lambda_s\leq\lambda_w\leq \lambda_v$. Putting the equations together, we get $\lambda_s \leq \lambda_v \leq \lambda_t$
\end{proof}

Next, we define a more natural variant for our setting, the sink thin flow $(x, \mu)$. 
Observe that we normalize $x$ by a factor of $\frac{\mu_v+1}{\mu_s+1}$. We do this because later on, we show that for all flow carrying vertices, $\mu_v = \frac{\lambda_t}{\lambda_v}-1$, so $\lambda_t = \mu_s+1$, and then $\lambda_v = \frac{\mu_s+1}{\mu_v+1}$. In other words, we scale $x_e$ for $e=vw$ by $\frac{1}{\lambda_v}$. In the non-tolled setting, this scaled value can easily be seen to correspond to the inflow rate into $e$ at time $\ell_v(\theta)$. 
Also, in our setting $x_e$ can be thought of the edge inflow rate, and $\mu_v$ as the derivative of the cost of a particle for reaching the sink.


\begin{definition}
   Let $(x,\mu)$ be a pair with $x \in \mathbb{R}_{\geq 0}^E$ and $\mu \in \mathbb{R}_{\geq 0}^V$. We say $(x,\mu)$ is a sink thin flow if $\left(\frac{x_e(\mu_v+1)}{\mu_s+1}\right)_{e=vw \in E}$ is an $s$-$t$-flow of value $u$ and  
\begin{align}
    \mu_t&=0, \tag{t-TF1} \label{eq:ttf1}\\
    \mu_v&=\min_{e=vw\in E}\max\{\mu_w,(1+\mu_w)\cdot x_e/\nu_e-1\}\text{ for all }v\in V\setminus\{t\}, \tag{t-TF2} \label{eq:ttf2}\\
    \mu_v&=\max\{\mu_w,(1+\mu_w)\cdot x_e/\nu_e-1\}\text{ for all }e=vw\in E\text{ with }x_e>0 \tag{t-TF3} \label{eq:ttf3}.
\end{align} 
\end{definition}

Source thin flows and sink thin flows are not equivalent. To characterize their difference, a useful distinction is between vertices the flow travels through, and the vertices through which it does not.  Formally, we define the set of {\em flow-carrying vertices} by $V^+\coloneqq\{v\in V \mid \exists\: e=vw\in E\text{ with }y_e>0\text{, or }\exists\: e=uv\in E\text{ with }y_e>0\}$.

\begin{lemma}
\label{lem:source_sink_thin_flow}
    Let $(y,\lambda)$ be a source thin flow. Then there exists a sink thin flow $(x,\mu)$, with $x_e=y_e/\lambda_v$ for all $e=vw\in E$ and $\mu_v=\lambda_t/\lambda_v-1$ for all $v\in V^+$.
\end{lemma}
\begin{proof}
    Note that $t \in V^+$ and thus by definition $\mu_t = \lambda_t / \lambda_t - 1 = 0$, matching \eqref{eq:ttf1}.

We now prove that \eqref{eq:ttf3} holds for all flow carrying edges. Let $e=vw\in E$ with $x_e>0$. Observe that $x_e>0$ if and only if $y_e>0$. From \eqref{eq:stf3}, $\lambda_w = \max\{\lambda_v, y_e/\nu_e\}$. 

We first consider the case: $\lambda_v \geq y_e/\nu_e$. Then $\lambda_v=\lambda_w \geq y_e/\nu_e=\lambda_vx_e/\nu_e$, and thus $x_e/\nu_e \leq 1$. Therefore, $(1+\mu_w)\cdot x_e/\nu_e - 1 \leq \mu_w$, so as $\max\{\mu_w,(1+\mu_w)\cdot x_e/\nu_e-1\}= \mu_w$. We still have to show that $\mu_v =\mu_w$. As both $v$ and $w$ are in $V^+$, by definition,
\[\mu_v = \lambda_t/\lambda_v-1=\lambda_t/\lambda_w - 1 = \mu_w.\]

We are left with the case: $\lambda_v<y_e/\nu_e$. Since $y_e/\nu_e=\lambda_vx_e/\nu_e$, by definition $x_e/\nu_e> 1$. Additionally, from \eqref{eq:stf3} $\lambda_w=y_e/\nu_e$. Therefore $(1+\mu_w)x_e/\nu_e-1>\mu_w$, so as $\max\{\mu_w,(1+\mu_w)\cdot x_e/\nu_e-1\} = (1+\mu_w)x_e/\nu_e-1$. We are left to show that $\mu_v=(1+\mu_w)x_e/\nu_e-1$. We have that
\[
    \mu_v = \lambda_t/\lambda_v - 1=x_e\lambda_t/y_e - 1=x_e\lambda_t/(\lambda_w\nu_e) - 1=(1+\mu_w)\cdot x_e/\nu_e-1,
\]
where the first equality follows by definition as $v \in V^+$, the second equality follows from $x_e=y_e/\lambda_v$, the third equality from $\lambda_w=y_e/\nu_e$, and the fourth equality from $\mu_w = \lambda_t/\lambda_w - 1$ which holds by definition as $w \in V^+$. 

We next prove that \eqref{eq:ttf2} holds for all $v\in V^+\setminus\{t\}$. Consider an edge $e = vw\in E$ with $x_e=0$, so by definition $y_e=0$. Therefore, $\max\{\lambda_v, y_e/\nu_e\} = \lambda_v$, so from \eqref{eq:stf2}, we know that $\lambda_w \leq \lambda_v$. As $(1+\mu_w)x_e/\nu_e -1 = -1 < \mu_w$, we use this to prove that $\mu_v \leq \mu_w$. By definition (as $v \in V^+$ by assumption):
\[\mu_v =  \lambda_t/\lambda_v - 1 \leq \lambda_t/\lambda_w - 1=\mu_w.\]

Note that for vertices $v \in V^+$, the minimum is met from \eqref{eq:stf3}. In order to complete the proof, we must first define a labeling for vertices $v \in V \setminus V^+$ and then show that for every $v\in V\setminus V^+$ there exists $e=vw$  such that $\mu_v = \max \{\mu_w, (1+\mu_w)\cdot x_e/\nu_e -1\}$. Given $v\in V\setminus V^+$, consider the set of simple $v$-$t$ paths $P_{vt}$. For each path $p \in P_{vt}$ consider the first flow-carrying vertex $u_p\in V^+$ on the path $p$. Note that as $t \in V^+$, $u_p$ is well defined. Then, define $\mu_v=\min_{p \in P_{vt}}\mu_{u_{p}}$. By construction, \eqref{eq:ttf2} is satisfied.
\end{proof}

Note that \Cref{thm:existence_thin_flows} and \Cref{lem:source_sink_thin_flow} (respectively, its proof) imply the following corollary. 
\begin{corollary}
    A sink thin flow always exists and can be computed in polynomial time.
\end{corollary}

Moreover, \Cref{lem:source_sink_thin_flow} shows a clear one-to-one-correspondence between the labels of flow-carrying vertices in source and sink thin flows. 

\begin{example}

    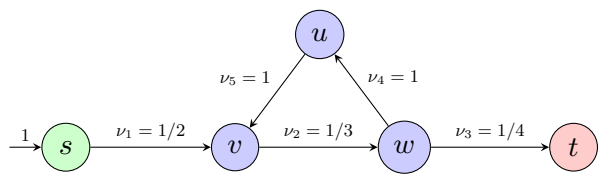
\begin{figure}[h]
\centering
\begin{tikzpicture}[scale=0.75]
\node[draw,circle,scale=1.25,fill=green!20] (1) at (0,0) {$s$};
\node[draw,circle,scale=1.2,fill=blue!20] (2) at (3,0) {$v$};
\node[draw,circle,scale=1.2,fill=blue!20] (3) at (6,0) {$w$};
\node[draw,circle,scale=1.2,fill=red!20] (4) at (9,0) {$t$};
\node[draw,circle,scale=1.2,fill=blue!20] (5) at (4.5,2) {$u$};
\draw[-stealth] (-1,0) to node[above,scale=0.75]{$1$} (1);
\draw[-stealth] (1) to node[above,scale=0.75]{$\nu_1=1/2$} (2);
\draw[-stealth] (2) to node[above,scale=0.75]{$\nu_2=1/3$}  (3);
\draw[-stealth] (3) to node[above,scale=0.75]{$\nu_3=1/4$}  (4);
\draw[-stealth] (3) to node[above right,scale=0.75]{$\nu_4=1$} (5);
\draw[-stealth] (5) to node[above left,scale=0.75]{$\nu_5=1$} (2);
\end{tikzpicture}
\caption{Source and sink thin flows. Transit times and tolls are not given as they do not affect the thin flows.}\label{fig:tf}
\end{figure}
Consider the instance given by Figure \ref{fig:tf}, in particular $E=\{1,2,3,4,5\}$. Consider the pair $(y,\lambda)$ with $y_1=y_2=y_3=1$, $y_4=y_5=0$, $\lambda_s=1$, $\lambda_v=2$, $\lambda_w=\lambda_u=3$ and $\lambda_t=4$. Then $(y,\lambda)$ is a source thin flow.
Similarly, the pair $(x,\mu)$ with $x_1=1$, $x_2=1/2$, $x_3=1/3$, $x_4=x_5=0$, $\mu_t=0$, $\mu_w=1/3$, $\mu_v=\mu_u=1$ and $\mu_s=3$ is a sink thin flow. 

We make the following two observations: (1) $\mu_u\neq \lambda_t/\lambda_u-1$. This fits our definitions, as $u$ is not a flow-carrying vertex. (2) $\mu$ represents the derivative of the cost of a particle for reaching the sink. While this second observation is not true for general flows over time, we show in Section \ref{sec:steady_state} that it is true for flow-carrying vertices in the flows that we define.
\end{example}

\section{Properties of Nash Flows over Time with Tolls}
\label{sec:examples}

We provide two examples to highlight the difference between dynamic equilibria with and without tolls. For the toll-free model, it is known that dynamic equilibria are unique with respect to travel times \cite{CCL15,OSK22} and reach a steady state at a certain point in time \cite{CCO22,OSK22}. We first show that dynamic equilibria with tolls need not be unique and then show that dynamic equilibria with tolls need not reach a steady state.

\begin{figure}[h]
\centering
\begin{minipage}[b]{0.48\textwidth}
\centering
\begin{tikzpicture}[scale=0.75]
\node[draw,circle,scale=1.25,fill=green!20] (1) at (0,0) {$s$};
\node[draw,circle,scale=1.2,fill=blue!20] (2) at (3,0) {$v$};
\node[draw,circle,scale=1.2,fill=red!20] (3) at (6,0) {$t$};
\draw[-stealth] (-1,0) to node[above,scale=0.75]{$2$} (1);
\draw[-stealth,bend left=30] (1) to node[above,scale=0.75]{$\tau_1=0, \nu_1=2,c_1=2$} (2);
\draw[-stealth,bend right=30] (1) to node[below,scale=0.75]{$\tau_2=1, \nu_2=2,c_2=0$} (2);
\draw[-stealth] (2) to node[above,scale=0.75]{$\tau_3=0, \nu_3=1$} node[below,scale=0.75]{$c_3=0$} (3);
\end{tikzpicture}
\caption{Multiple dynamic equilibria with different cost functions.}
\label{fig:mul}
\end{minipage}
\hfill
\begin{minipage}[b]{0.48\textwidth}
\centering
\begin{tikzpicture}[scale=0.75]
\node[draw,circle,scale=1.25,fill=green!20] (1) at (0,0) {$s$};
\node[draw,circle,scale=1.2,fill=blue!20] (2) at (3,0) {$v$};
\node[draw,circle,scale=1.2,fill=red!20] (3) at (6,0) {$t$};
\draw[-stealth] (-1,0) to node[above,scale=0.75]{$2$} (1);
\draw[-stealth,bend left=30] (1) to node[above,scale=0.75]{$\tau_1=0, \nu_1=1,c_1=1$} (2);
\draw[-stealth,bend right=30] (1) to node[below,scale=0.75]{$\tau_2=1, \nu_2=1,c_2=0$} (2);
\draw[-stealth] (2) to node[above,scale=0.75]{$\tau_3=0, \nu_3=1$} node[below,scale=0.75]{$c_3=0$} (3);
\end{tikzpicture}
\caption{A Nash flow over time that never reaches steady state.}
\label{fig:ss}
\end{minipage}
\end{figure}

For Nash flows over time, we cannot hope to show uniqueness of the actual flow in the network, even without tolls. To see this consider two parallel edges both with capacity equal to the inflow. Here, flow can split arbitrary between both edges. However, the transit time of each particle is the same in every equilibrium. In our model this translates to the hope that the cost functions of the particles could be the same across equilibria. However, the following theorem shows this is not case.

\begin{theorem}
\label{thm:non_uniqueness}
    A Nash flow over time is not unique with respect to the cost functions.
\end{theorem}
\begin{proof}
Consider the instance depicted in Figure~\ref{fig:mul} with inflow rate $u=2$.

We now provide two different equilibria for this instance. 
First, consider the flow over time $f$ where all particles travel along the upper path. Formally, we define $f_1^+(\theta)=f_3^+(\theta)=2$ and $f_2^+(\theta)=0$ for all $\theta\geq 0$.
Then the corresponding cost functions are $c_v(\theta)=q_3(\theta)=\theta$ and $c_s(\theta)=\min\{c_1+c_v(\theta),\tau_2+c_v(\theta+\tau_2)\}=\min\{2+\theta,1+\theta +1\}=2+\theta$.
As both $sv$ edges attain the minimum in the definition of $c_2$ for every point in time all edges in this example are always active and thus $f$ is a dynamic equilibrium.

Consider the flow over time $g$ which sends all the flow along the bottom path. Formally, we define $g_1^+(\theta)=0$, $g_2^+(\theta)=2$, and \[g_3^+(\theta)=\begin{cases}
    0&\text{ if }\theta<1,\\
    2&\text{ if }\theta\geq 1,
\end{cases}\]Then  \[c_v(\theta) = q_3(\theta)=\begin{cases}
    0 \text{ for all }\theta<1,\\
    \theta-1\text{ for all }\theta\geq1,
\end{cases}\] and \[c_s(\theta)=\min\{c_1+c_v(\theta),\tau_2+c_v(\theta+\tau_2)\}=\begin{cases}
    \min\{2,1+1+\theta-1\}=1+\theta\text{ for all }\theta<1,\\
    \min\{2+\theta-1,1+1+\theta-1\}=1+\theta\text{ for all }\theta\geq1.
\end{cases}\]

Hence for $\theta<1$ only the bottom edge between $s$ and $v$ is active and then from time 1 on both edges are active. This yields, that $g$ is also a dynamic equilibrium. Observe that $f$ has different induced costs than $g$.

This example illustrates the difference between tolls and travel times. Users are willing to pay a higher toll than travel time so as to skip waiting time in a queue later in the network. 
\end{proof}

 In the proof we give two different Nash flows over time. However, we can show that every constant flow over time different from the two above (for example, equal split over the two paths) is not a dynamic equilibrium. Thus, the set of dynamic equilibria is not convex.


The next theorem provides an instance in which the dynamic equilibrium consists of infinitely many phases. It converges towards a steady state, but never reaches it.

\begin{theorem}
    \label{thm:no_steady_states_reached}
A Nash flow over time with tolls does not always reach a steady state and can have infinitely many phases.
\end{theorem}

\begin{proof}
We consider the instance depicted in Figure~\ref{fig:ss} with inflow rate $u=2$. We remark that the result is more general and we found that the instance does not reach a steady state for all tolls $c_1$ with $0<c_1<8/5$. 

We define a Nash flows over time with tolls that does not reach a steady state. For all $i\in\mathbb{N}\cup\{0\}$, define $\theta_i\coloneqq\frac{(2i+1)\cdot 3^i-1}{4\cdot 3^i}$. 
We define the piece-wise constant flow $f$ as follows. For all $i\in \mathbb{N}\cup\{0\}$ and for all $\theta\in [\theta_i, \theta_{i+1})$, we have
\begin{align*}
    f^+_1(\theta)&=\frac{3^{i+1}+3}{3^{i+1}+1},\\
    f^+_2(\theta)&=f^+_3(1+\theta) - 1=\frac{3^{i+1}-1}{3^{i+1}+1},
\end{align*}
and $f^+_3(\theta)=1$ for all $\theta\in[0,1)$.

Note that $f_1(\theta)\geq 1$ for all $\theta\geq0$ and $\theta_i$ is monotonously increasing, and goes to infinity as $i$ goes to infinity. One can check that flow-conservation is fulfilled.
Thus, the defined flow is indeed a feasible flow over time. The induced queues at time $\theta_i$ can be computed as follows for all $i\in\mathbb{N}\cup\{0\}$,
\begin{align*}
    q_1(\theta_i)&=\sum_{j=0}^{i-1}(\theta_{j+1}-\theta_j)\cdot (f^+_1(\theta_j)-1)=\sum_{j=0}^{i-1} \frac{3^{j+1}+1}{2\cdot 3^{j+1}}\cdot\frac{2}{3^{j+1} + 1}=\sum_{j=0}^{i-1} \frac{1}{3^{j+1}}=\frac{3^i-1}{2\cdot 3^i},\\
q_3(1+\theta_i)&=\sum_{j=0}^{i-1}(\theta_{j+1}-\theta_j)\cdot (f^+_3(1+\theta_j)-1)=\sum_{j=0}^{i-1} \frac{3^{j+1}+1}{2\cdot 3^{j+1}}\cdot\frac{3^{j+1}-1}{3^{j+1} + 1}=\frac{(2i-1)\cdot 3^i+1}{4\cdot 3^i}.
\end{align*}
We are now ready to prove that $f$ is a dynamic equilibrium.

Observe that $c_v(\theta)=q_3(\theta)$. We show that $c_s(\theta)=q_1(\theta)+c_1+c_v(\theta+q_1(\theta))=\tau_2+c_v(\theta+\tau_2)$ for all $\theta=\theta_i$ and for all $i\in\mathbb{N}\cup\{0\}$, which implies that both paths have equal costs to reach $t$. The result then follows since flows are piecewise constant and thus costs are piecewise linear.

Notice that for all $i\in \mathbb{N}$, \begin{align*}
\theta_i+q_1(\theta_i)&=\frac{(2i+1)\cdot 3^i-1}{4\cdot 3^i}+\frac{3^i-1}{2\cdot 3^i}=\frac{(2i+3)\cdot 3^i -3}{4\cdot 3^i}\\
&=\frac{4\cdot 3^{i-1} +(2(i-1)+1)\cdot 3^{i-1} -1}{4\cdot 3^{i-1}}=1+\frac{(2(i-1)+1)\cdot 3^{i-1}-1}{4\cdot 3^{i-1}}=1+\theta_{i-1}.\end{align*}

This implies that a particle departing at time $\theta_i$ using edge $e_1$ arrives together with a particle departing at time $\theta_{i-1}$ using edge $e_2$ and thus face the same queue on edge $e_3$. This implies that 
\begin{align*}
c_1+q_1(\theta_i)+c_v(\theta_i+q_1(\theta_i))&=c_1+q_1(\theta_i)+q_3(\theta_i+q_1(\theta_i))\\
&=1+q_1(\theta_i)+q_3(1+\theta_{i-1})\\
&=1+\frac{3^i-1}{2\cdot 3^i}+\frac{(2(i-1)-1)\cdot 3^{i-1}+1}{4\cdot 3^{i-1}}\\
&=1+\frac{(2i-1)\cdot 3^i+1}{4\cdot 3^i}\\
&=1+q_3(1+\theta_i) \\
&= \tau_2+c_v(1+\theta_i)
\end{align*}

Thus, $c_s(\theta)$ is defined as claimed. In particular, this implies that both $sv$ edges remain active throughout the evolution of the flow over time. This shows that $f$ is indeed an equilibrium flow.
\end{proof}

 \begin{figure}[h]
  \centering
  \includegraphics[width=\linewidth]{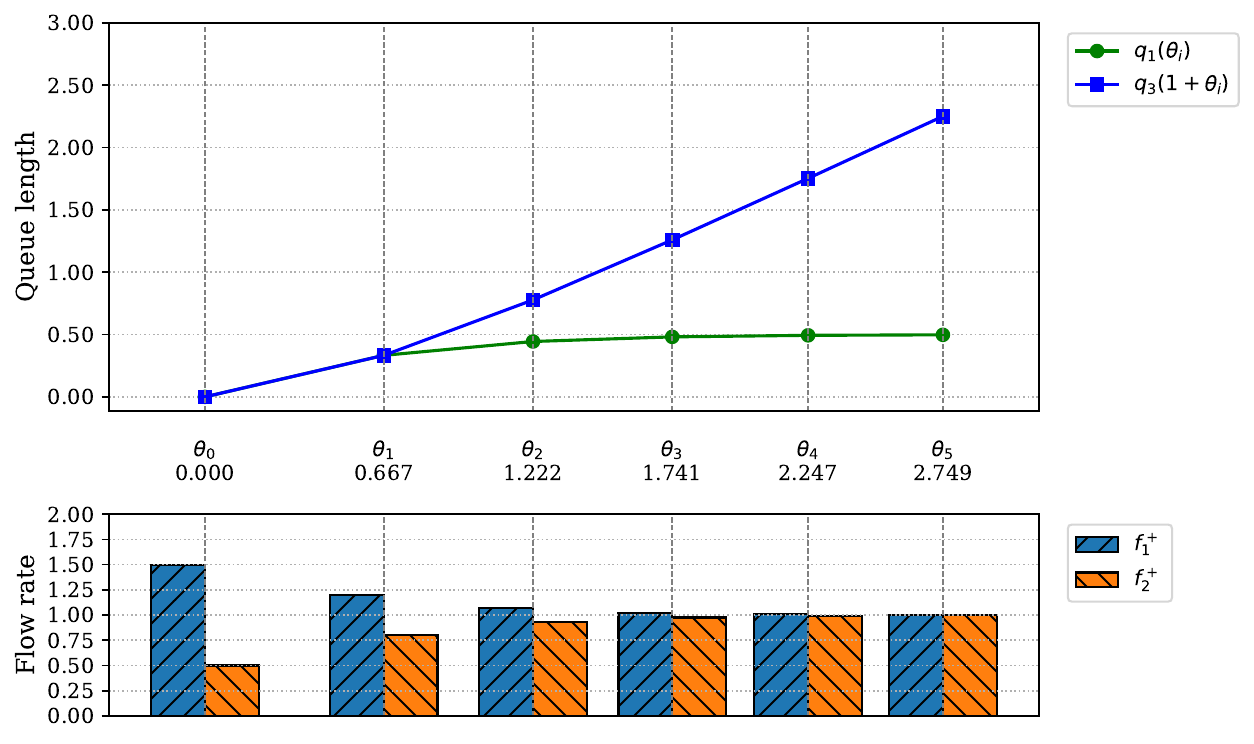}
  \caption{Flow division and queue lengths for different $\theta_i$ values for the example in \Cref{thm:no_steady_states_reached}. It can be observed that queue length as well as flow rates converge rather quickly, while the explicit formulas show that the steady state is never actually reached.}
  \label{fig:flow-plot}
\end{figure}

\section{Computing the Steady State}
\label{sec:steady_state}
In the previous section, we have seen an instance where an equilibrium flow starting in the initially empty network never reaches a steady state (\Cref{thm:no_steady_states_reached}). Recall that in a steady state, from some point in time onward all queues change with a constant rate, i.e., linearly.
 In this section, we describe how to compute steady states using a linear program, and how to implement them by setting initial queues, for a variant of flow over time which we denote as a {\em deficit flow}. 
 The deficit flow differs from a standard flow over time, in that every vertex is treated as a source until some time $T$, but from time $T$ (which for our purpose can be set equal to the maximum transit time of an edge) onward we will show that it is a feasible flow over time, which only uses active edges, and sends a constant amount of flow along all these edges. Note that this usage of deficit flows is necessary, in the sense that by \Cref{thm:no_steady_states_reached} a steady state may not be reached from an initial network state.

For this section, we consider a fixed instance of a tolled flow over time, which is given by $(G=(V,E), u,  (\tau_e)_{e \in E}, (\nu_e)_{e \in E}, (c_e)_{e \in E})$. 
A \emph{deficit flow until time $T$ with initial queues $q$} is a set of edge inflow and outflow rates $f=(f_e^+,f_e^-)_{e\in E}$, together with a vector of initial queues $q \in \mathbb{R}^E_{\geq 0}$ such that the edge inflow and outflow rates satisfy the following constraints. 
First, 
\[f_e^-(\theta+\tau_e)=\begin{cases}
\nu_e &\text{ if }z_e(\theta)>0 \; ,\\
\min\{f_e^+(\theta),\nu_e\} &\text{ if }z_e(\theta)=0 \;,
\end{cases}\]
for all $e\in E$ and all $\theta\geq0$. Here, the definition of the queue mass $z$ has to be adapted slightly to take the starting queue into account, namely $z_e(\theta) = F_e^+(\theta) - F_e^-(\theta+\tau_e) + \nu_e q_e$.
Second, for all $v\in V$ and $\theta\geq0$
\[\sum_{e\in\delta^+(v)}f_e^+(\theta)-\sum_{e \in \delta^-(v)}f_e^-(\theta)\begin{cases}
= u &\text{ if }v=s \; ,\\
\leq 0 & \text{ if } v=t \; ,\\
\geq 0 &\text{ if }v\neq s,t \; ,\\
= 0 &\text{ if }v\neq s,t \; \text{ and } \theta \geq T.
\end{cases}\]

We will use this definition, to allow for an initial phase in which the network fills up, before the flow has to correspond to a proper flow over time. 
To see that the definition of a deficit flow with initial queues is quite natural when we aim to describe steady state flows, we consider a simple example.
\begin{example}
    Consider a network with vertices $V=\{s,v,t\}$ and $e_1=sv$ and $e_2=e_3=vt$, as depicted in \Cref{fig:deficit_flow}. Let $\tau_{e_2}=2$ and $\tau_{e_1}=\tau_{e_3}=1$, $u=\nu_{e_1}=2$ and $\nu_{e_2}=\nu_{e_3}=1$. All tolls are set to 0. 
    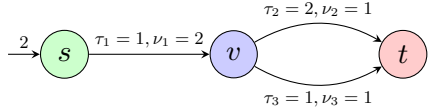
\begin{figure}[h]
\centering
\begin{tikzpicture}[scale=0.75]
\node[draw,circle,scale=1.25,fill=green!20] (1) at (0,0) {$s$};
\node[draw,circle,scale=1.2,fill=blue!20] (2) at (3,0) {$v$};
\node[draw,circle,scale=1.2,fill=red!20] (3) at (6,0) {$t$};
\draw[-stealth] (-1,0) to node[above,scale=0.75]{$2$} (1);
\draw[-stealth,bend left=30] (2) to node[above,scale=0.75]{$\tau_2=2, \nu_2=1$} (3);
\draw[-stealth,bend right=30] (2) to node[below,scale=0.75]{$\tau_3=1, \nu_3=1$} (3);
\draw[-stealth] (1) to node[above,scale=0.75]{$\tau_1=1, \nu_1=2$} node[below,scale=0.75]{ } (2);
\end{tikzpicture}
\caption{All tolls in the above example are 0. With a deficit flow steady states can be described in a simple way.}\label{fig:deficit_flow}
\end{figure}
    It is easy to see that in a steady state there should be a queue of length $1$ on edge $e_3$. Once such a queue forms, the flow should split half-half between the top path and the bottom path. However, if we initially place a queue of length $1$ on $e_3$ and start to send flow into the network with rate $2$, the queue will start to deplete immediately. The deficit flow allows us to send flow originating at $v$, in order to keep the queue constant until the flow from the source reaches the end of $e_1$. This allows for a more natural assignment of queues which does not try to anticipate the depletion of the queues that happens until the network is filled up. In this sense a deficit flow is quite natural. Here, a deficit flow until time $1$ with initial queues $q_{e_3}=1$ and $q_{e_1}=q_{e_2}=0$ is given as $f^+_{e_1}(\theta)=2 =f_{e_1}^-(\theta+1)$, 
    $f_{e_2}^+(\theta)=f_{e_3}^+(\theta)=1=f_{e_2}^-(\theta+2)=f_{e_3}^-(\theta+1)$ for all $\theta$. Moreover, $f_{e_i}^-(\theta) = 0$ for $\theta \leq \tau_{e_i}$ for $i \in [3]$.
\end{example}
We extend the definitions of dynamic equilibrium and steady state to deficit flows in the natural way.

In this section, we use deficit flows in order to define a steady state for a flow over time with tolls. Formally, we prove that there is a deficit flow until time $\taumax\coloneqq\max_{e\in E} \tau_e$ with a set of initial queues such that $q'_e(\theta)$ is constant for every $e \in E$ and $\theta \geq 0$.
Formally:
\begin{theorem}
\label{thm:steady}
    There exists a deficit flow $f$ until time $\taumax$ with initial queues $(q_e)_{e \in E}$, such that
    \begin{enumerate}
        \item $f_e^+$ is constant for all $e\in E$,
        \item  $f$ is in steady state from time $0$ onward,
        \item $f$ is a dynamic equilibrium.
    \end{enumerate}
\end{theorem}




Our proof of Theorem \ref{thm:steady} is comprised of two steps. First, we introduce a primal-dual pair of linear problems based on a thin flow $(y, \lambda)$. We prove that there is an optimal solution to these linear programs that corresponds to a deficit flow, denoted as the {\em LP-induced flow}. Next, we prove that this LP induced flow is a dynamic equilibrium. 

\subsection{LP Representation}
\label{sec:LPs}
Let $(\bar y,\lambda)$ be a source thin flow for the given tolled flow over time instance. This thin flow exists by \Cref{thm:existence_thin_flows}. 
In this section, we introduce a primal-dual pair of linear programs based on $\lambda$ in order to define our steady state. We will then prove that there is an optimal solution to these linear programs that satisfies several useful properties.
The linear program as well as the notation we introduce in the following are strongly inspired by \cite{OSK22}, where a linear program was presented to characterize steady states in flows over time without tolls. This linear program was again inspired by the work of \cite{CCO22}. We will adapt the linear program to the tolled setting.


We first characterize the edges according to the value of $\lambda$:

\begin{center}
    $E^>:=\{vw \in E \mid \lambda_w > \lambda_v$\}\\
    $E^=:=\{vw \in E \mid \lambda_w = \lambda_v$\}\\
    $E^<:=\{vw \in E \mid \lambda_w < \lambda_v$\}\\
\end{center}

Note that this partitions the edge set, i.e. $E= E^> \dot\cup E^= \dot\cup E^<$.

Consider the following minimum flow LP. This is similar to the LP presented in \cite{OSK22} with two differences: \begin{enumerate}
    \item The primal objective function and thus the right hand side of the dual constraint is modified ($\frac{\lambda_t}{\lambda_w} \tau_e + c_e$ replaces $\tau_e$)
    \item The constraint $y_e=0$ for $e \in E^<$ is removed, which yields an additional dual constraint. This turns out to be useful for our proof, and does not change the optimal primal solution.
\end{enumerate}
\begin{alignat}{2}
    \min \quad & \sum_{e=vw \in E} \left( \frac{\lambda_t}{\lambda_w} \tau_e + c_e \right) y_e & \quad & \\
    \text{s.t.} \quad & \text{$y$ is an $s$-$t$-flow of value } u & & \\\label{eq:pri}
    & y_e = \lambda_w \nu_e & & \quad \forall e = vw \in E^> \\
    & y_e \leq \lambda_w \nu_e & & \quad \forall e = vw \in E^= \cup E^< \\
    & y_e \geq 0 & & \quad \forall e \in E
\end{alignat}

Its dual is:
\begin{alignat}{2}
    \max \quad & u (d_s - d_t) - \sum_{e=vw \in E} \lambda_w \nu_e p_e & \quad & \\\label{eq:dual}
    \text{s.t.} \quad & d_v - d_w - p_e \leq \frac{\lambda_t}{\lambda_w} \tau_e + c_e & & \quad \forall e = vw \in E \\\label{eq:dual2}
    & p_e \geq 0 & & \quad \forall e \in E^= \cup E^<
\end{alignat}

We first prove that $(y, \lambda)$ is a source thin flow for every feasible solution $y$ to the LP.

\begin{lemma}\label{lem:stf}
   Let the $s$-$t$-flow $y$ be a feasible solution to the LP. Then $(y,\lambda)$ is a source thin flow.
\end{lemma}
\begin{proof}
    Observe that from the primal constraints $$\max\{\lambda_v,y_e/\nu_e\}=\begin{cases}
        \lambda_w&\text{ if }e\in E^>,\\
        \lambda_v=\lambda_w&\text{ if }e\in E^=,\\
        \lambda_v&\text{ if }e\in E^<,
    \end{cases}$$
    which implies that $\lambda_w \leq \max \{ \lambda_v, y_e/\nu_e\}$ for all $e\in E$.
   
    We prove that $(y,\lambda)$ is a source thin flow. 
    First, recall that there exists a $\bar y$ such that $(\bar y, \lambda)$ is a source thin flow. Thus, clearly $\lambda_s=1$, i.e., \eqref{eq:stf1} is satisfied. 
    
    By the observation, we have that $\eqref{eq:stf3}$ is satisfied for all $e\in E^>\cup E^=$. To show $\eqref{eq:stf3}$ is satisfied for all $e\in E$, we will prove that $y_e=0$ for all $e\in E^<$.

    Assume that $\hat e=vw \in E^<$. Then by the observation, $\lambda_w<\lambda_v=\max\{\lambda_v,y_e/\nu_e\}$. Define $S=\{u \in V \mid \lambda_u \leq \lambda_w\}$. 
    By \Cref{lem:lam}, we get that $(S,V\setminus S)$ is an $s$-$t$ cut. Moreover, by definition $\hat e \in \delta^-(S)$. 
    As $\bar y$ is an $s$-$t$ flow we know from \eqref{eq:stf3} that $\bar y_e=0$ for all $e\in \delta^-(S)$, so we have that 
    \begin{align*}\sum_{e\in\delta^+(S)}\bar y_e-\sum_{e\in\delta^-(S)}\bar y_e=\sum_{e \in \delta^+(S)}\bar{y}_e=\sum_{e\in\delta^+(S)}\lambda_w\nu_e=u.
    \end{align*}
    Where the second equality is from $\eqref{eq:stf3}$.
    Since all edges $e=vw \in \delta^+(S)$ are in $E^>$, from the primal constraints $y_e=\lambda_w\nu_e$. As $y$ is also an $s$-$t$ flow, we have that
    \begin{align*}
    u = \sum_{e\in\delta^+(S)}y_e - \sum_{e\in\delta^-(S)}y_e=\sum_{e\in\delta^+(S)}\lambda_w\nu_e - \sum_{e\in\delta^-(S)} y_e= u -  \sum_{e\in\delta^-(S)} y_e.
    \end{align*}
    This implies that $y_e=0$ for all $e\in\delta^-(S)$ and thus in particular for $\hat e$.

    We are left to show $\eqref{eq:stf2}$ is satisfied. First, for every vertex $w$ with positive inflow with respect to $y$, this is directly implied from $\eqref{eq:stf3}$, as for edges with $y_e=0$ we have $e\in E^=\cup E^<$ and thus $\lambda_w\leq \lambda_v$. 
    
    We now prove $\eqref{eq:stf2}$ is satisfied for a vertex $w\neq s$ without any positive inflow with respect to $y$. For every edge $e=vw \in E$, we know $y_e=0$. We already observed that $\lambda_w \leq \lambda_v= \max \{ \lambda_v, y_e/\nu_e\}$ for all $e=vw \in E$. So assume, contradictory to what we want to prove, that $\lambda_{w}<\lambda_{v}$ for all $e=v w\in E$. Then $\bar y_{e}=0$ for all $e=vw \in E$ from $\eqref{eq:stf3}$ for $(\bar{y}, \lambda)$. 
    This implies that $\max \{ \lambda_{v}, y_{ e}/\nu_{ e}\}=\lambda_{ v}$ for all $e= vw\in E$ and thus $\lambda_w<\min_{e=vw\in E}\lambda_v=\min_{e=vw\in E}\max \{ \lambda_{v}, y_{ e}/\nu_{e}\}$, which contradicts $\eqref{eq:stf2}$ for $(\bar{y}, \lambda)$.
\end{proof}

A particularly useful fact that the above result implies is that $y_e=0$ for all $e\in E^<$.

In order to describe feasible steady states, we prove that there is always an optimal dual solution $(d^*,p^*)$ that satisfies the following two additional constraints.

\begin{lemma}
\label{lem:dual_struct}
    There is an optimal dual solution $(d^*, p^*)$ for which:
\begin{enumerate}
    \item $d^*_t = 0$.
    \item $p^*_e \geq 0$ for every $e \in E$.
\end{enumerate}
\end{lemma}
\begin{proof}
We first show that the dual LP is feasible and bounded. We do this by observing that both the primal and the dual LP are feasible (and using duality). 
The dual is feasible as the all 0 vector is a feasible solution. The primal LP is feasible as the source thin flow, which exists from \Cref{thm:existence_thin_flows}, satisfies all constraints.


Now, let $(d, p)$ be some optimal solution to the dual LP.

    We first prove that there is an optimal solution $(d^*, p)$ such that $d^*_t = 0$. Define $d^*_v=d_v-d_t$ for all $v \in V$, and observe that $d^*_t=0$. The objective of the dual LP is unchanged as $d^*_s - d^*_t = (d_s - d_t) - (d_t - d_t) = d_s - d_t$. Additionally, the solution remains feasible, as for every $e=vw \in E$, $d^*_v - d^*_w = (d_v - d_t) - (d_w - d_t) = d_v - d_w$. 
    
    We next prove that there is an optimal solution $(d^*, p^*)$ to the dual LP such that $p^*_e \geq 0$ for all $e \in E$. If $p_e\geq 0$ for every $e \in E$, we are done. Otherwise, let $e^-=vw$ be an edge such that $p_{e^-} <0$. Note that it must be that $e^- \in E^>$, so $\lambda_v < \lambda_w$.

    Define $S=\{u \in V \mid \lambda_u < \lambda_w\}$. By \Cref{lem:lam}, we get that $(S,V\setminus S)$ is an $s$-$t$ cut. Moreover, observe that $e^- \in \delta^+(S)$.

    We are now ready to define $(d^*,p^*)$:
    $$p^*_e=\begin{cases}
p_e-p_{e^-}&\text{ if }e \in \delta^+(S)\\
p_e&\text{ otherwise.}
\end{cases}$$

and $$d^*_u=\begin{cases}
d_u-p_{e^-}&\text{ if }u \in S\\
d_u&\text{ otherwise.}
\end{cases}$$

We show that $(d^*, p^*)$ is an optimal solution of the dual LP. First, we show that $(d^*, p^*)$ is feasible. For all $e=\bar v \bar w \notin \delta^+(S) \cup \delta^-(S)$, we have that $p^*_e=p_e$. Moreover, either $\bar v, \bar w \in S$ or $\bar v, \bar w \not\in S$. In either case $d_{\bar v}-d_{\bar w}=d_{\bar v}^*-d_{\bar w}^*$  and thus \eqref{eq:dual} is  satisfied. 
For all $e=\bar v \bar w \in \delta^+(S)$, we have that $\bar v \in S$ and $\bar w \in V \setminus S$ and thus $d^*_v - d^*_w - p^*_e = d_v - p_{e^-} - d_w - (p_e - p_{e^-}) = d_v -d_w -p_e$, so \eqref{eq:dual} is satisfied.
For all $e=\bar v \bar w \in \delta^-(S)$, we have that $\bar v \not\in S$ and $\bar w \in S$ and thus 
\[d^*_v - d^*_w - p^*_e = d_v  - (d_w - p_{e^-})- p_e\leq  d_v  - d_w - p_e \leq \frac{\lambda_t}{\lambda_w} \tau_e +c_e.\]

Observe that $\eqref{eq:dual2}$ is clearly satisfied.

Second, we show that $(d^*, p^*)$ is optimal. The objective of the dual LP is 
\begin{align*}
&u(d^*_s - d^*_t) - \sum_{e =vw \in E}\lambda_w\nu_e p_e^* \\
&= u(d_s-p_{e^-}-d_t) - \sum_{e = vw \in E} \lambda_w\nu_ep_e + \sum_{e \in \delta^+(S)}\lambda_w\nu_ep_{e^-} \\
&= u(d_s-p_{e^-}-d_t) - \sum_{e = vw \in E} \lambda_w\nu_ep_e + \sum_{e \in \delta^+(S)}y^*_ep_{e^-} \\
&= u(d_s-d_t) - \sum_{e = vw \in E} \lambda_w\nu_ep_e + up_{e^-} - up_{e^-}\\ 
&=  u(d_s-d_t) - \sum_{e = vw \in E} \lambda_w\nu_ep_e,
\end{align*}
where $y^*$ is an optimal primal solution. The second equality follows from $\eqref{eq:pri}$ in the primal constraints, and the third equality from $y^*$ being an $s$-$t$ flow, $(S,V\setminus S)$ being an $s$-$t$ cut, $e\in\delta^-(S)$ implying $e\in E^<$ and thus by \Cref{lem:stf}, $y^*_e=0$ for all $e\in\delta^-(S)$, together implying that $\sum_{e\in \delta^+(S)}y^*_e= u$.

Thus, $(d^*,p^*)$ is an optimal solution of the dual LP. By repeating this argument for all $e^-$ with $p_{e^-} <0$, we are left with an optimal solution for which $p_e^* \geq 0$ for all $e \in E$.
\end{proof}

\subsection{LP Induced Flow}
\label{sec:dynamic_eq}
We will now transform optimal solutions of the LP and its dual, to a deficit flow over time with initial queues which are used to prove Theorem \ref{thm:steady}. Let $y^*$ and $(d^*,p^*)$ be an optimal primal dual pair, satisfying the restrictions of \Cref{lem:dual_struct}. 
We now describe edge inflow functions and initial queues corresponding to the pair $y^*$ and $(d^*,p^*)$, which we denote by \emph{LP induced flow}, and that we will show to satisfy the three properties of Theorem \ref{thm:steady}. Define 
\begin{equation}
\label{eq:LPflow}
f^+_e(\theta) = \frac{y^*_e}{\lambda_v}
\quad \forall e=vw\in E, \,\,\,\, \forall \theta\geq0 \qquad \text{ and } \qquad q_e=\frac{\lambda_w}{\lambda_t}\cdot p^*_e \quad \forall e=vw\in E.
\end{equation}
  By \Cref{lem:lam}, this is well-defined and by \Cref{lem:dual_struct}, all initial queues are non-negative. Observe that from complementary slackness $p_e^*=0$ for all $e\in E$ with $y_e<\lambda_w \nu_e$. By defining $f^-$ as the outflow rates induced by the queues and $f^+$, we obtain the LP-induced flow $f=(f^+,f^-)$.

\subsubsection{Deficit flow}
We start by showing that $f$ is a deficit flow with initial queues $q_e$.
\begin{lemma}
\label{lem:deficit}
    An LP-induced flow $f$ is a deficit flow until $\taumax$ with initial queues $(q_e)_{e \in E}$ .
\end{lemma}
\begin{proof}
    Note that the definition of $f_e^+$ is feasible, because $\lambda_s=1$ and $\lambda_v\geq 1$ for all $v \in V$. Since $f^+_e(\theta) = \frac{y_e^*}{\lambda_v}$ for all $e=vw\in E$ and for all $\theta\geq0$, $$f^-_e(\theta+\tau_e)=\min\{f_e^+(\theta),\nu_e\}=\begin{cases}
        \nu_e&\text{ if }e\in E^>,\\
        y_e^*/\lambda_w&\text{ if }e=vw\in E^=,\\
        0&\text{ if }e\in E^<,
    \end{cases}$$ 
    
    for all $e\in E$ and for all $\theta\geq0$. Recall that from complementary slackness $p_e^*=0$ for all $e\in E$ with $y^*_e<\lambda_w \nu_e$. This implies that all edges with an initial non-zero queue have an outflow of $\nu_e$ and an inflow of at least $\nu_e$.
    Moreover, we can deduce that for every $\theta\geq \max\{\tau_e\mid e\in E^>\cup E^=\}$ (and thus in particular for every $\theta \geq \taumax$), for every $w \in V$,
    \begin{align*}
        \sum_{e\in\delta^+(w)}f_e^+(\theta)-\sum_{e \in \delta^-(w)}f_e^-(\theta)&=\sum_{e\in\delta^+(w)}\frac{y^*_e}{\lambda_w}-\sum_{e \in \delta^-(w)\cap E^>}\nu_e-\sum_{e \in \delta^-(w)\cap E^=}\frac{y^*_e}{\lambda_w}\\
        &=\sum_{e\in\delta^+(w)}\frac{y^*_e}{\lambda_w}-\sum_{e \in \delta^-(w)\cap E^>}\frac{y^*_e}{\lambda_w}-\sum_{e \in \delta^-(w)\cap E^=}\frac{y^*_e}{\lambda_w}\\
        &=0.\end{align*}

        Where the first equality is from the definition of the LP-induced flow, and the second equality is from the primal constraints.

        For $\theta \leq \max\{\tau_e\mid e\in E^>\cup E^=\}\leq \taumax$ there can be terms of $f_e^-(\theta)$ which are 0. As all terms are non-negative it follows that in this case vertices  can have a positive net outflow, which is allowed in a deficit flow.
\end{proof}

\subsubsection{Steady state}
We proceed to prove that $f$ is in steady state from time $0$ onward.

\begin{lemma}
\label{lem:q_deriv}
For an LP induced flow, $q_e'(\theta) =\frac{\lambda_w}{\lambda_v} -1$ for all $e=vw\in E^>$, and $q_e'(\theta) =0$ for all $e\in E^=\cup  E^<$. 
\end{lemma}
\begin{proof}
We distinguish two cases. If  $e\in E^=$ and $y_e=\lambda_w \nu_e$ or $e\in E^>$, then:
\begin{align*}
 q_e'(\theta) = \frac{f_e^+(\theta) - \nu_e }{\nu_e}=\frac{\frac{\nu_e \lambda_w}{\lambda_v} - \nu_e }{\nu_e}=\frac{\lambda_w}{\lambda_v} - 1.
\end{align*}

For edges $e\in E^=$ such that $y^*_e<\lambda_w \nu_e$, $p^*_e=0$ from complementary slackness. For edges $e \in E^<$, we showed $y^*_e=0<\lambda_w \nu_e$, so $p_e^*=0$ from complementary slackness as well. Thus, since $f_e^+(\theta)<\nu_e$, $q_e'(\theta)=0$ for both these types of edges.
\qed
\end{proof}
Note that this implies $q_e'(\theta)\geq 0$ for all $e \in E$ and $\theta \geq 0$. Moreover, clearly queues just change linearly from time 0 on, meaning the flow is in steady state. 

\subsubsection{Dynamic equilibrium}
We are now ready to prove the last property: that the LP-induced flow is a dynamic equilibrium. Recall that a flow $f$ is a dynamic equilibrium if whenever an edge $e$ has positive inflow at time $\theta$, then $e$ is active at that time, i.e. $e \in E_{\theta}$.

We provide an overview of the steps of this proof. First, for all $v\in V$ and all paths $P$ from $v$ to $t$, we define $c^P_v(\theta)$ as the cost to reach $t$ starting from $v$ at time $\theta$ using path $P$. We use this definition to deduce that the rate of increase of cost on flow-carrying paths is minimal among all paths, and is equal to $\mu_v=\lambda_t/\lambda_v - 1$.

Next, we show that at time $0$, the cost on flow-carrying paths is minimal among all paths, and is equal to $d_v^*$. 

Lastly, we combine these claims on the minimality of $c_v(0)$ and $(c_v^P)'$ for paths of flow-carrying edges to deduce that they are minimum cost paths.



Note that in the following lemma, we overload the notation of the path $P$, to also denote the set of vertices used in the path.

    \begin{lemma}
    \label{lem:derivatives}
    For an LP induced flow, for each $v\in V$, for each $v$-$t$ path $P$
    \[(c^{P}_{v})'(\theta) \geq\\
    \lambda_t/\min_{u \in P}\lambda_u-1\geq \\
    \lambda_t/\lambda_v-1\] Moreover, if $e\in E^=\cup E^>$ for all $e\in P$, then 
    \[(c^{P}_{v})'(\theta) = \lambda_t/\lambda_v-1\].
    \end{lemma}

    \begin{proof}
    We prove this by induction over the number of edges in the path. Assume that $P$ consists of a single edge $e$. By \Cref{lem:lam}, $e=vt \in E^=\cup E^>$. 
    Since by definition $c^P_v(\theta)=\tau_e+q_e(\theta)+c_e$, we have that
    \begin{align*}(c^{P}_{v})'(\theta)=q'_e(\theta)=\lambda_t/\lambda_v-1,
    \end{align*}
    where the second equality follows from \Cref{lem:q_deriv}. This proves the base case for both parts of the lemma.

    We now prove 
    \begin{equation}
    \label{eq:c_derivative}
        (c^{P}_{v})'(\theta) \geq \lambda_t/\min_{u \in P}\lambda_u-1
    \end{equation} 
    for paths $P$ of any length. This clearly implies $(c^{P}_{v})'(\theta) \geq \lambda_t/\lambda_v-1$ (as $v \in P$). 
    Assume \eqref{eq:c_derivative} holds true for all $v\in V$ and all $v$-$t$ paths consisting of $k\geq 1$ edges. Let $P$ be a $v$-$t$ path consisting of $k+1$ edges where the first edge $e\in P$ is $e=vw$. Let $P_{wt}$ denote the path $P \setminus \{e\}$.
    Observe that $c^P_v(\theta)=\tau_e+q_e(\theta)+c_e+c^{P_{wt}}_w(\theta+\tau_e+q_e(\theta))$. We distinguish between two cases: $e\in E^<$, and $e\in E^=\cup E^>$. If $e\in E^<$ we have that
    \begin{align*}
    (c^{P}_{v})'(\theta)&=q'_e(\theta)+(c^{P_{wt}}_w)'(\theta+\tau_e+q_e(\theta))\cdot (1+q'_e(\theta))=(c^{P_{wt}}_w)'(\theta+\tau_e)\\
    &\geq\lambda_t/\min_{u \in P_{wt}}\lambda_u-1=\lambda_t/\min_{u \in P}\lambda_u-1,
    \end{align*}
    where the second equality follows from \Cref{lem:q_deriv}, which shows that for $e \in E^<$, $q_e(\theta)=0$ and $q'_e(\theta)=0$. The inequality follows from the induction hypothesis, and the third equality follows from $e\in E^<$ which implies $\min_{u \in P}\lambda_u \neq \lambda_v$.

    If $e\in E^=\cup E^>$, we have that
    \begin{align*}
    (c^{P}_{v})'(\theta)&=q'_e(\theta)+(c^{P_{wt}}_w)'(\theta+\tau_e+q_e(\theta))\cdot (1+q'_e(\theta))\\&=\left(1+(c^{P_{wt}}_w)'(\theta+\tau_e+q_e(\theta))\right)\cdot (\lambda_w/\lambda_v-1)+(c^{P_{wt}}_w)'(\theta+\tau_e+q_e(\theta))\\
    &\geq\lambda_t/\min_{u \in P_{wt}}\lambda_u \cdot (\lambda_w/\lambda_v-1)+\left(\lambda_t/\min_{u \in P_{wt}}\lambda_u -1\right)\\
    &=\lambda_t\cdot\lambda_w/\left(\min_{u \in P_{wt}}\lambda_u \cdot \lambda_v\right)-1\\
    &\geq \lambda_t/\min_{u \in P}\lambda_u -1,
    \end{align*}
    where the second equality follows from \Cref{lem:q_deriv} which yields $q'_e(\theta)=\lambda_w/\lambda_v-1$, the first inequality follows from the induction hypothesis, and the second inequality follows from $\lambda_w\geq \min_{u \in P_{wt}}\lambda_u $ if $\min_{u \in P}\lambda_u =\lambda_v$ and $\lambda_w\geq \lambda_v$ if $\min_{u \in P}\lambda_u \neq \lambda_v$.

    We now prove the second statement using a similar induction claim. Assume that $(c^{P}_{v})'(\theta)=\lambda_t/\lambda_v-1$ for all $v\in V$ and all $v$-$t$ paths $P$ with $e\in E^=\cup E^>$ for all $e\in P$ consisting of $k\geq 1$ edges. Let $P$ be a $v$-$t$ path with $e\in E^=\cup E^>$ for all $e\in P$ consisting of $k+1$ edges where the first edge $e\in P$ is $e=vw$. We have that
    \begin{align*}
    (c^{P}_{v})'(\theta)&=q'_e(\theta)+(c^{P_{wt}}_w)'(\theta+\tau_e+q_e(\theta))\cdot (1+q'_e(\theta))\\&=\left(1+(c^{P_{wt}}_w)'(\theta+\tau_e+q_e(\theta))\right)\cdot (\lambda_w/\lambda_v-1)+(c^{P_{wt}}_w)'(\theta+\tau_e+q_e(\theta))\\
    &=\lambda_t/\lambda_w\cdot (\lambda_w/\lambda_v-1)+\left(\lambda_t/\lambda_w-1\right)\\
    &=\lambda_t/\lambda_v-1,
    \end{align*}
    where the second equality follows from \Cref{lem:q_deriv} and thus $q'_e(\theta)=\lambda_w/\lambda_v-1$ and the third equality follows from the induction hypothesis.
    \qed
    \end{proof}


We next consider the cost at time 0. For an LP induced flow, denote by $E^+\coloneqq\{e\in E\mid f^+_e>0\}$ the set of flow-carrying edges (recall that $f_e^+$ is constant).

\begin{lemma}
\label{lem:cost_zero_to_right}
  For an LP induced flow, for each $v \in V$, for each $v$-$t$ path $P$, $c_v^P(0) \geq d^*_v$. Moreover, if $e\in E^+$ for all $e\in P$, then $c_v^P(0)=d_v^*$.
\end{lemma}

\begin{proof}

 We prove the two statements by induction over the number of edges in the path. Assume that $P$ consists of a single edge $e=vt$. Then
$c^P_v(0)=\tau_e+q_e(0)+c_e=\tau_e+\frac{\lambda_t}{\lambda_t}p^*_e+c_e = \frac{\lambda_t}{\lambda_t}\tau_e + p^*_e +c_e \geq d^*_v$, where the second equality follows by definition of an LP induced flow, the third equality since $d^*_t=0$ from Lemma~\ref{lem:dual_struct}, and the inequality follows from \eqref{eq:dual} in the dual. By complementary slackness, the inequality is an equality if $f^+_e>0$.

We now prove the statement for paths $P$ of any length. Assume that $c_v^P(0) \geq d^*_v$ for all $v\in V$ and all $v$-$t$ paths $P$ consisting of $k\geq 1$ edges, with equality if $e\in E^+$ for all $e\in P$. Let $P$ be a $v$-$t$ path consisting of $k+1$ edges where the first edge $e\in P$ is $e=vw$. Let $P_{wt}$ denote the path $P \setminus \{e\}$. Then
    \begin{align*}
    c^P_v(0)&=\tau_e+q_e(0)+c_e+c^{P_{wt}}_w(\tau_e+q_e(0))\\
    &\geq \tau_e+q_e(0)+c_e+c^{P_{wt}}_w(0)+(\tau_e+q_e(0))\cdot (\lambda_t/\lambda_w-1)\\
    &\geq(\lambda_t/\lambda_w)\tau_e+p^*_e+c_e+d^*_w\\
    &\geq d^*_v,
    \end{align*}
    where the first inequality follows from \Cref{lem:derivatives}, the second inequality from the definition of an LP induced flow and the induction hypothesis, and the third inequality from \eqref{eq:dual}. Observe that three inequalities are equalities if $e\in E^+$ and thus $e\in E^>\cup E^=$, where the last equality uses complementary slackness.
    \qed
    \end{proof}

It now follows from \Cref{lem:cost_zero_to_right} that every flow-carrying edge is active at time $0$. By \Cref{lem:derivatives}, we then have that every flow-carrying edge is active at every time $\theta>0$.
    \begin{corollary}
    \label{cor:flow_carrying_path_active}
        An LP induced flow is a dynamic equilibrium.
    \end{corollary}

Now the proof of \Cref{thm:steady} is straight forward.
\begin{proof}[of \Cref{thm:steady}]
Consider an LP induced flow $f$, with initial queues $q$ as defined in Equation~\eqref{eq:LPflow}. By \Cref{lem:deficit}, $f$ is a deficit flow with initial queues $q$. By definition, $f_e^+$ is constant for every $e \in E$. By \Cref{lem:derivatives}, $f$ is in steady state from time $0$ onward. By \Cref{cor:flow_carrying_path_active}, $f$ is a dynamic equilibrium.
\qed
\end{proof}

\section{Discussion}
\label{sec:discussion}
In this paper we augmented the existing deterministic fluid queueing model by introducing tolls. While tolls have been previously considered in the context of implementing specific flows as equilibria, to the best of our knowledge this is the first work attempting to understand the behavior of equilibrium flows under given constant tolls.


In this work, we aim at understanding the long term behavior of the model. We provide an example showing that convergence to steady state is not guaranteed, as a dynamic equilibrium might have infinitely many phases without reaching a steady state. Moreover, we describe how to construct an equilibrium flow in steady states.

Still basic questions for this model remain open and will benefit from further research. First, do dynamic equilibria always exist in our model? Based on current existence results, existence is only guaranteed if the inflow is positive for a fixed amount of time. Second, do dynamic equilibria still consist of phases, i.e., is a dynamic equilibrium piecewise constant? Third, can we say something about the efficiency of equilibria in this model?

An issue with extending the thin flow characterization of the toll-free model, is that with tolls, the travel time to arrive at a vertex $v$ for different particles originating at the same time $\theta$ can be different. In order to overcome this issue, we proposed an alternative sink thin flow definition. Can a definition along these lines be used to give a thin flow like characterization of Nash flows over time with tolls?

An additional shortcoming of our results is our limitation to single commodity instances and constant tolls.
Single-commodity refers to two things, first that all particles start at the same source and travel to the same sink and second that all particles have the same trade off between money and time.
An extension of our results in any of these directions would be a significant contribution.

\section*{Acknowledgement}
We want to thank Wenzel Manegold for finding the example of Theorem 3.1, and Lukas Graf and Neil Olver for fruitful discussion on the topic of this work. 
We also thank the participants of the Dagstuhl Seminar 24281 \cite{correa2025dynamic} for helpful discussions and questions.  

    \bibliography{refs}

\begin{thebibliography}{10}
\providecommand{\url}[1]{\texttt{#1}}
\providecommand{\urlprefix}{URL }
\providecommand{\doi}[1]{https://doi.org/#1}

\bibitem{beckmann}
Beckmann, M.J., McGuire, C.B., Winsten, C.B.: Studies in Economics of
  Transportation. Yale University Press, New Haven, Connecticut (1956)

\bibitem{cole}
Cole, R., Dodis, Y., Roughgarden, T.: Pricing network edges for heterogeneous
  selfish users. In: Proceedings of the thirty-fifth annual ACM symposium on
  Theory of computing. pp. 521--530 (2003)

\bibitem{CCL15}
Cominetti, R., Correa, J., Larr{\'e}, O.: Dynamic equilibria in fluid queueing
  networks. Operations Research  \textbf{63}(1),  21--34 (2015)

\bibitem{CCO22}
Cominetti, R., Correa, J., Olver, N.: Long-term behavior of dynamic equilibria
  in fluid queuing networks. Operations Research  \textbf{70}(1),  516--526
  (2022)

\bibitem{correa2025dynamic}
Correa, J., Osorio, C., Koch, L.V., Watling, D., Griesbach, S.: Dynamic traffic
  models in transportation science (dagstuhl seminar 24281). Dagstuhl Reports
  \textbf{14}(7),  1--16 (2025)

\bibitem{ford1958constructing}
Ford~Jr, L.R., Fulkerson, D.R.: Constructing maximal dynamic flows from static
  flows. Operations research  \textbf{6}(3),  419--433 (1958)

\bibitem{frascaria2020algorithms}
Frascaria, D., Olver, N.: Algorithms for flows over time with scheduling costs.
  In: Integer Programming and Combinatorial Optimization: 21st International
  Conference, IPCO 2020, London, UK, June 8--10, 2020, Proceedings. pp.
  130--143. Springer (2020)

\bibitem{gairing}
Gairing, M., Paccagnan, D.: In congestion games, taxes achieve optimal
  approximation. Operations Research  \textbf{72}(3),  966--982 (2023)

\bibitem{GH23}
Graf, L., Harks, T.: Side-constrained dynamic traffic equilibria. In:
  Proceedings of the 24th ACM Conference on Economics and Computation. pp.
  814--814 (2023)

\bibitem{graf2025system}
Graf, L., Harks, T., Schwarz, J.: Are system optimal dynamic flows
  implementable by tolls? arXiv preprint arXiv:2503.07387  (2025)

\bibitem{graf2025tolls}
Graf, L., Harks, T., Schwarz, J.: Tolls for dynamic equilibrium flows. In:
  Proceedings of the 2025 Annual ACM-SIAM Symposium on Discrete Algorithms
  (SODA). pp. 2560--2606. SIAM (2025)

\bibitem{iryo2011multiple}
Iryo, T.: Multiple equilibria in a dynamic traffic network. Transportation
  Research Part B: Methodological  \textbf{45}(6),  867--879 (2011)

\bibitem{iryo2017uniqueness}
Iryo, T., Smith, M.J.: On the uniqueness of equilibrated dynamic traffic flow
  patterns in unidirectional networks. Transportation research procedia
  \textbf{23},  283--302 (2017)

\bibitem{Koc12}
Koch, R.: Routing Games over Time. Ph.{D}. thesis, TU Berlin, Berlin, Germany
  (2012)

\bibitem{koch2011nash}
Koch, R., Skutella, M.: Nash equilibria and the price of anarchy for flows over
  time. Theory of Computing Systems  \textbf{49}(1),  71--97 (2011)

\bibitem{OSK22}
Olver, N., Sering, L., Vargas~Koch, L.: Continuity, uniqueness and long-term
  behavior of {N}ash flows over time. In: 2021 IEEE 62nd Annual Symposium on
  Foundations of Computer Science (FOCS). pp. 851--860. IEEE (2022)

\bibitem{OSS}
Oosterwijk, T., Schmand, D., Schr{\"o}der, M.: Bicriteria {N}ash flows over
  time. Games and Economic Behavior  \textbf{147},  19--37 (2024)

\bibitem{pigou}
Pigou, A.: The economics of welfare. Macmillan (1920)

\bibitem{sering2019nash}
Sering, L., Vargas~Koch, L.: Nash flows over time with spillback. In:
  Proceedings of the Thirtieth Annual ACM-SIAM Symposium on Discrete
  Algorithms. pp. 935--945. SIAM (2019)

\bibitem{vickrey1969congestion}
Vickrey, W.S.: Congestion theory and transport investment. The American
  economic review  \textbf{59}(2),  251--260 (1969)

\bibitem{ziemke2023spillback}
Ziemke, T., Sering, L., Nagel, K.: Spillback changes the long-term behavior of
  dynamic equilibria in fluid queuing networks. In: 23rd Symposium on
  Algorithmic Approaches for Transportation Modelling, Optimization, and
  Systems (ATMOS 2023). pp. 11--1. Schloss Dagstuhl--Leibniz-Zentrum f{\"u}r
  Informatik (2023)

\end{thebibliography}
\end{document}